\documentclass[pra,twocolumn,showpacs,superscriptaddress]{revtex4-1}

\usepackage{amsmath,amssymb,amsthm,mathrsfs}

\usepackage[matrix,frame,arrow]{xy}
\usepackage{amsmath}
\newcommand{\bra}[1]{\left\langle{#1}\right\vert}
\newcommand{\ket}[1]{\left\vert{#1}\right\rangle}
\newcommand{\qw}[1][-1]{\ar @{-} [0,#1]}

\newcommand{\measureD}[1]{*{\xy*+=+<.5em>{\vphantom{\rule{0em}{.1em}#1}}*\cir{r_l};p\save*!R{#1} \restore\save+UC;+UC-<.5em,0em>*!R{\hphantom{#1}}+L **\dir{-} \restore\save+DC;+DC-<.5em,0em>*!R{\hphantom{#1}}+L **\dir{-} \restore\POS+UC-<.5em,0em>*!R{\hphantom{#1}}+L;+DC-<.5em,0em>*!R{\hphantom{#1}}+L **\dir{-} \endxy} \qw}

\newcommand{\multigate}[2]{*+<1em,.9em>{\hphantom{#2}} \qw \POS[0,0].[#1,0];p !C *{#2},p \save+LU;+RU **\dir{-}\restore\save+RU;+RD **\dir{-}\restore\save+RD;+LD **\dir{-}\restore\save+LD;+LU **\dir{-}\restore}
\newcommand{\ghost}[1]{*+<1em,.9em>{\hphantom{#1}} \qw}

\newcommand{\ustick}[1]{*!D!<0em,-.5em>=<0em>{#1}}

\newcommand{\Qcircuit}[1][0em]{\xymatrix @*=<#1>}

\newcommand{\pureghost}[1]{*+<1em,.9em>{\hphantom{#1}}}
\newcommand{\multiprepareC}[2]{*+<1em,.9em>{\hphantom{#2}}\save[0,0].[#1,0];p\save !C
  *{#2},p+RU+<0em,0em>;+LU+<+.8em,0em> **\dir{-}\restore\save +RD;+RU **\dir{-}\restore\save
  +RD;+LD+<.8em,0em> **\dir{-} \restore\save +LD+<0em,.8em>;+LU-<0em,.8em> **\dir{-} \restore \POS
  !UL*!UL{\cir<.9em>{u_r}};!DL*!DL{\cir<.9em>{l_u}}\restore}
\newcommand{\prepareC}[1]{*{\xy*+=+<.5em>{\vphantom{#1\rule{0em}{.1em}}}*\cir{l^r};p\save*!L{#1} \restore\save+UC;+UC+<.5em,0em>*!L{\hphantom{#1}}+R **\dir{-} \restore\save+DC;+DC+<.5em,0em>*!L{\hphantom{#1}}+R **\dir{-} \restore\POS+UC+<.5em,0em>*!L{\hphantom{#1}}+R;+DC+<.5em,0em>*!L{\hphantom{#1}}+R **\dir{-} \endxy}}

\newcommand{\kket}[1]{|#1\rangle\rangle}

\newcommand{\eq}[1]{Eq. (\ref{#1})}

\def\Tr{\operatorname{Tr}} \def\span{\operatorname{span}}
\def\rank{\operatorname{rank}} \def\supp{\operatorname{supp}}

\newtheorem{thm}{Proposition} 

\begin{document}

\title{Quantum error correction with degenerate codes for correlated noise}

\author{Giulio \surname{Chiribella}}

\affiliation{Perimeter Institute for Theoretical Physics, 31 Caroline
  St. North, Waterloo, Ontario N2L 2Y5, Canada}

\author{Michele \surname{Dall'Arno}} \author{Giacomo Mauro \surname{D'Ariano}}
\author{Chiara \surname{Macchiavello}} \author{Paolo \surname{Perinotti}}

\affiliation{Quit group, Dipartimento di Fisica ``A. Volta'', via Bassi 6,
  I-27100 Pavia, Italy} \affiliation{INFN Sezione di Pavia, via Bassi 6,
  I-27100 Pavia, Italy}

\date{\today}

\begin{abstract}
  We introduce a quantum packing bound on the minimal resources required by
  nondegenerate error correction codes for any kind of noise. We prove that
  degenerate codes can outperform nondegenerate ones in the presence of
  correlated noise, by exhibiting examples where the quantum packing bound is
  violated.
\end{abstract}

\maketitle

\section{Introduction}

Since its early development in 1995 \cite{shor95, steane96, gottesman96,
  lidar98, keyl02, gregoratti03, kribs04}, the theory of quantum error
correction has played a major role to design strategies for protecting quantum
information in the presence of noise. This task is particularly relevant in
several contexts, such as the communication of quantum information over
quantum channels \cite{NC00}, and fault tolerant quantum computation
\cite{AKP06}. In view of the massive experimental effort in investigating
suitable quantum computational systems and of future implementations, it is of
great importance to establish what are the minimum resources needed to have
successful quantum error correction.

A useful bound that allows to quantify them is the quantum Hamming bound
\cite{ekert96} for nondegenerate codes, namely codes for which each error is
individually identifiable. This bound holds when the dominant terms in the
noise process correspond to all the error operators that involve at most a
fixed number of subsystems. This is the case, for example, when the noise
affects independently every single subsystem (uncorrelated
noise). Intuitively, the quantum Hamming bound can be explained from the fact
that if each error is individually identifiable, then it must send the encoded
information into orthogonal subspaces, thus setting a lower bound on the
dimension of the system. Actually, so far no degenerate code has been proved
to violate the quantum Hamming bound \cite{gottesman09}, and for some classes
of degenerate codes the impossibility of violating the bound has been
demonstrated \cite{sarvepalli08}.

Here, we derive a general bound that we call quantum packing bound,
constraining the resource requirements for the correction of any kind of noise
process. In particular, the quantum Hamming bound is an instance of the
quantum packing bound for the case of arbitrary noise process affecting at most
a fixed number of systems.

The assumption of uncorrelated noise may not hold in many physical
implementations of fault tolerant quantum computers, such as ion traps
\cite{ion_traps}, quantum dots \cite{quantum_dots} or solid state systems
\cite{solid_state}. In this work we study the resource requirements of error
correction codes in the presence of correlated noise, namely noise processes
where perfect correlations among the encoding subsystems dominate and
therefore not all the strings of single particle noisy processes are
relevant. We show that degenerate codes can outperform nondegenerate ones in
the presence of correlated noise. The resource requirements of quantum error
correction codes for a particular class of correlated errors, namely spatially
correlated (or burst) errors, have been studied in \cite{VRA99}.

The paper is organized as follows. In Sect. \ref{sect:bound} we provide the
quantum packing bound on the minimal resources required by a nondegenerate
code in the presence of any kind of noise. We refer to such bound as quantum
packing bound. We prove that the quantum Hamming bound can be recovered as a
particular case. In Sect. \ref{sect:examples} we show that in the presence of
correlated noise the quantum packing bound can be violated by degenerate
codes, which lead to a much more compact transmission of information,
unveiling the fact that, for some noise channels, degenerate codes can work
better than nondegenerate codes in terms of resources required. Finally, we
summarize our results in Sect. \ref{sect:conclusion}.

\section{Quantum packing bound}\label{sect:bound}

An $[[n, k, d]]_q$ quantum error-correcting code is given by a
$q^k$-dimensional subspace of the state space $\mathscr{H} =
(\mathbb{C}^q)^{\otimes n}$ of $n$ quantum systems with $q$ levels where it is
possible to correct all errors affecting at most $(d-1)/2$ quantum systems.
We denote by $P_Q$ the projector onto the quantum code $Q$.  Let $S_E$ denote
a subspace of linear operators on $\mathscr H$.  The quantum code $Q$ is able
to correct all errors in $S_E$ if and only if there exists an Hermitian matrix
$M$ such that for any pair of error operators $L_i,L_j$ which belong to a
basis on $S_E$ \cite{knill97}
\begin{equation}\label{eq:knill_laflamme}
  P_Q L_i^\dagger L_j P_Q = M_{ij} P_Q.
\end{equation}
In the above expression, known as Knill-Laflamme correctability condition,
$M_{ij}$ are the entries of the Hermitian matrix $M$, that depends on the
choice of the $L_i$'s.  The pair $(Q, S_E)$, consisting of a quantum code $Q$
and a vector space of errors $S_E$, is called {\em degenerate} if and only if
the Hermitian matrix $M$ in \eq{eq:knill_laflamme} is singular; otherwise,
$(Q, S_E)$ is called {\em nondegenerate}.  In other words, in the case of
nondegenerate codes and with a suitable choice of error operators, the quantum
code is transformed into a set of distinct orthogonal subspaces by applying
the error operators, while for degenerate codes it may happen that distinct
error operators transform the code into the same subspace.

We derive now a bound for nondegenerate codes, which does not depend on the
form of the noise acting on the encoding system, and which can be reduced to
the well known quantum Hamming bound \cite{ekert96} as a particular case. To
this purpose, we first give a brief overview on a few results about error
correction.  Let us denote the {\em system} by $S$ and the {\em encoding
  subspace} by $Q \subseteq S$.  Given a state $\rho^S$, we say that a channel
$\cal{E}$ is \emph{correctable upon input of $\rho^S$} if and only if there
exists a channel $\cal{R}$ such that $\cal R \cal E (\sigma) = \sigma$ for
every $\sigma$ with $\supp(\sigma) \subseteq \supp(\rho^S)$. Clearly, if
$\supp(\rho^S) = Q$, this is equivalent to say that $Q$ is a good quantum
code.

Let us introduce a purification $\rho^{SR}$ of $\rho^S$, $R$ being the
\emph{reference}.  It is easy to see that for arbitrary channels $\cal C$ and
$\cal D$ we have ${\cal I}_R \otimes {\cal C} (\rho^{SR}) = {\cal I}_R \otimes
\cal{D} (\rho^{SR})$ if and only if $\cal{C} (\sigma) = \cal{D} (\sigma)$ for
any $\sigma$ with $\supp(\sigma) \subseteq \supp(\rho^S)$.  This fact implies
that $\cal E$ is correctable upon input of $\rho$ if and only if we have
${\cal I}_R \otimes {\cal{RE}} (\rho^{SR}) = \rho^{SR}$.  Taking a unitary
dilation $(U_{\cal{E}}, \ket \eta)$ [$(U_{\cal{R}},\ket\xi$)] of channel
$\cal{E}$ ($\cal R$) with {\em environment} $E$ ($A$), this is equivalent to
the following equation:
\begin{equation}\label{eq:pur00}
  \begin{aligned}
    \Qcircuit @C=2mm @R=2mm{ \multiprepareC{1}{\rho^{SR}} & \ustick{R} \qw &
      \qw & \qw & \qw & \qw & \qw \\ \pureghost{\rho^{SR}} & \ustick{S} \qw &
      \multigate{2}{U_{\cal{E}}} & \qw & \qw & \multigate{1}{U_{\cal{R}}} &
      \qw \\ & & \pureghost{U_{\cal{E}}} & \prepareC{\ket{\xi}} & \ustick{A}
      \qw & \ghost{U_{\cal{R}}} & \measureD{I} \\ \prepareC{{\ket\eta}} &
      \ustick{E} \qw & \ghost{U_{\cal{E}}} & \qw & \qw & \qw & \measureD{I}
      \\ }
  \end{aligned}
  =
  \begin{aligned}
    \Qcircuit @C=2mm @R=2mm { \multiprepareC{1}{\rho^{SR}} & \ustick{R} \qw &
      \qw \\ \pureghost{\rho^{SR}} & \ustick{S} \qw & \qw
      \\ \prepareC{\ket{\xi}} & \ustick{A} \qw & \measureD{I}
      \\ \prepareC{\ket{\eta}} & \ustick{E} \qw & \measureD{I} }
  \end{aligned},
\end{equation}
where $\Qcircuit @C=2mm @R=2mm{\pureghost{} & \measureD{I}}$ represents the
partial trace.  Consider the circuit on the left hand side. Denote by
$\rho^{S'R'E'}$ the state of $S$, $R$, and $E$ after the action of
$U_{\cal{E}}$, and by $\rho^{S'R'}$ ($\rho^{R'E'}$) its marginal on $SR$
($RE$). Since $\rho^{S'R'E'}$ is a purification for $\rho^{R'E'}$ we have
\begin{equation}\label{eq:dim00}
  \dim(S) \ge \rank(\rho^{R'E'}).
\end{equation}

We now give two other necessary and sufficient conditions for correctability,
which imply the quantum packing bound.

\begin{thm}\label{thm:uncorrelated}
  A channel $\cal E$ is correctable upon input of $\rho^S$ if and only if the
  reference $R$ and the environment $E$ are uncorrelated after the
  interaction, i.e. $\rho^{R'E'} = \rho^{R'} \otimes \rho^{E'}$ (see,
  e.g. \cite{SW02,CDP10}).
\end{thm}

\begin{proof}
  We repeat here only the proof of necessity since it is the only part needed.
  Calling $\rho^{S''R''E''A''}$ the state of $SREA$ after the action of
  $U_{\cal E}$ and $U_{\cal R}$, Eq. (\ref{eq:pur00}) is nothing but the
  statement that $\rho^{S''R''E''A''}$ and $\rho^{SR} \otimes \eta^E\otimes
  \xi^F$ are both purifications of $\rho^{RS}$. Therefore, there exists a
  unitary $U_{\cal{P}}$ such that
  \begin{equation}\label{pur01}
    \begin{aligned}
      \Qcircuit @C=2mm @R=2mm{ \multiprepareC{1}{\rho^{SR}} & \ustick{R} \qw &
        \qw & \qw & \qw & \qw & \qw \\ \pureghost{\rho^{SR}} & \ustick{S} \qw
        & \multigate{2}{U_{\cal{E}}} & \qw & \qw & \multigate{1}{U_{\cal{R}}}
        & \qw \\ & & \pureghost{U_{\cal{E}}} & \prepareC{\ket{\xi}} &
        \ustick{A} \qw & \ghost{U_{\cal{R}}} & \qw \\ \prepareC{\ket{\eta}} &
        \ustick{E} \qw & \ghost{U_{\cal{E}}} & \qw & \qw & \qw & \qw \\ }
    \end{aligned}
    =
    \begin{aligned}
      \Qcircuit @C=2mm @R=2mm{ \multiprepareC{1}{\rho^{SR}} & \ustick{R} \qw &
        \qw & \qw \\ \pureghost{\rho^{SR}} & \ustick{S} \qw & \qw & \qw
        \\ \prepareC{\ket{\xi}} & \ustick{A} \qw & \multigate{1}{U_{\cal{P}}}
        & \qw \\ \prepareC{\ket{\eta}} & \ustick{E} \qw & \ghost{U_{\cal{P}}}
        & \qw }
    \end{aligned}.
  \end{equation}
  Discarding systems $S$ and $A$ on both sides one then obtains
  $\rho^{R'E'} = \rho^{R'} \otimes \rho^{E'}$.
\end{proof}

The second necessary and sufficient condition is:

\begin{thm}\label{thm:deletion}
  A channel $\cal E$ is correctable upon input of $\rho^S$ if and only if its
  complementary channel $\cal \tilde{E}$ - namely the channel from $S$ to $E'$
  obtained by tracing $S'$ instead of $E'$ - is a deletion channel upon input
  of $\rho^S$, i.e. $ \tilde {\cal E}(\sigma) = \rho^{E'}$ for every $\sigma$
  with ${\rm supp} (\sigma) \subseteq {\supp} (\rho^S)$ \cite{kks} (see also
  \cite{CDP10} for a graphical proof).
\end{thm}

\begin{proof}
  We reproduce here only the proof of necessity because it is the only part
  needed for our considerations.  Eq. (\ref{pur01}) implies that for every
  $\sigma$ with ${\rm supp} (\sigma) \subseteq {\rm supp} (\rho^S)$ we have
  \begin{equation}\label{pur02}
    \begin{aligned}
      \Qcircuit @C=2mm @R=2mm{ \prepareC{\sigma} & \ustick{S} \qw &
        \multigate{2}{U_{\cal{E}}} & \qw & \qw & \multigate{1}{U_{\cal{R}}} &
        \qw\\ & & \pureghost{U_{\cal{E}}} & \prepareC{\ket{\xi}} & \ustick{A}
        \qw & \ghost{U_{\cal{R}}} & \qw \\ \prepareC{\ket{\eta}} & \ustick{E}
        \qw & \ghost{U_{\cal{E}}} & \qw & \qw & \qw & \qw }
    \end{aligned}
    =
    \begin{aligned}
      \Qcircuit @C=2mm @R=2mm{ \prepareC{\sigma} & \ustick{S} \qw & \qw & \qw
        \\ \prepareC{\ket{\xi}} & \ustick{A} \qw & \multigate{1}{U_{\cal{P}}}
        & \qw \\ \prepareC{\ket{\eta}} & \ustick{E} \qw & \ghost{U_{\cal{P}}}
        & \qw }
    \end{aligned}.
  \end{equation}
  Taking the partial trace over $S$ and $A$ on both sides we then obtain
  $\tilde {\cal E } (\sigma) = \rho^{E'}$, thus proving that $\tilde {\cal E}$
  is a deletion channel.
\end{proof}

The proofs of Prop. \ref{thm:uncorrelated} and \ref{thm:deletion} rely only
upon the very general requirement that any mixed state admits a unique
purification up to reversible transformations, thus holding for any
probabilistic theory with purification \cite{CDP10}.

Restricting now Prop. \ref{thm:uncorrelated} and \ref{thm:deletion} to the
quantum case, we derive the following bounds.
Using Prop. \ref{thm:uncorrelated}, \eq{eq:dim00} becomes
\begin{equation}\label{eq:dim01}
  \dim(S) \ge \rank(\rho^{R'})\rank(\rho^{E'}) =
  \rank(\rho^{S})\rank(\rho^{E'}),
\end{equation}
where last equality holds since $\cal{E}$ does not act on $R$ and $R$ purifies
$\rho^S$, so $\rank(\rho^{R'})=\rank(\rho^{R})=\rank(\rho^{S})$.
Proposition \ref{thm:deletion} allows us to identify the matrix $M$ in the
Knill-Laflamme condition (\ref{eq:knill_laflamme}) with (the transpose of)
$\rho^{E'}$. Indeed, for every (unnormalized) state of the form $P_Q \rho
P_Q$ the complementary channel $\tilde{\cal E}$ acts as
\begin{eqnarray}\label{eq:erasing_unitary}
  \tilde{\cal E}(P_Q \rho P_Q) &=& \Tr_S[U_{\cal E} (P_Q \rho P_Q \otimes
    \ket{\eta}\bra{\eta}) U_{\cal E}^\dagger] \\ &= &\Tr[P_Q \rho P_Q]
  \rho^{E'}, \nonumber
\end{eqnarray}
If $\cal E$ has Kraus decomposition $\mathcal{E}(\rho) = \sum_i L_i \rho
L_{i}^\dagger$ and $U_{\cal E}\ket{\eta} = \sum_i L_i \otimes \ket{e_i}$,
where $\ket{e_i}$ is an orthonormal set in $E$, then \eq{eq:erasing_unitary}
becomes
\begin{equation}\label{eq:erasing_krauss}
  \Tr_S[L_i P_Q \rho P_Q L_j^\dagger] = \Tr[P_Q \rho P_Q] \rho_{ij}^{E'}
\end{equation}
Taking $\rho = |\psi \rangle\langle\psi|$ with an arbitrary $\ket{\psi} \in
\mathscr H$, \eq{eq:erasing_krauss} becomes
\begin{equation}
  \bra{\psi} P_Q L_j^\dagger L_i P_Q\ket{\psi} = \bra{\psi} P_Q \ket{\psi}
  \rho_{ij}^{E'},
\end{equation}
which is equivalent to the Knill-Laflamme condition (\ref{eq:knill_laflamme})
with $\rho^{E'} = M^T$.
Summarizing, we proved the following result,
\begin{equation}
  \dim(S) \ge \rank(\rho^{S}) \rank(M),
\end{equation}
or, equivalently,
\begin{equation}\label{eq:bound}
  \dim (S) \ge \dim (Q) \rank (M).
\end{equation}
Notice that $\rank(M)$ does not depend on the choice of Kraus operators
$\{L_i\}$. In particular, to compute $\rank(M)$ we can use a minimal Kraus
decomposition $\mathcal{E}(\rho) = \sum_{i=1}^{\rank (R_{\cal E}) } K_i \rho
K_{i}^\dagger$, whose cardinality is equal to the rank of the Choi-Jamio\l
kowski operator $R_{\cal E} = ({\cal E} \otimes {\cal I}) (|I\rangle\!\rangle
\langle\!\langle I|)$, obtained by applying the channel $\cal E$ on one side
of the maximally entangled vector $|I\rangle\!\rangle = \sum_{n=1}^d
|n\rangle|n\rangle$.

We now consider the case of nondegenerate codes, i.e. codes for which the
matrix $M$ is non-singular.  In this case $\textrm{rank}(M)$ equals
$\rank(R_{\cal{E}})$, namely the cardinality of the minimal Kraus $\{K_i\}$.

\begin{thm}[Quantum packing bound]\label{thm:packing}
  Given a quantum channel $\cal{E}$ with Choi-Jamio\l kowski operator $R_{\cal
    E}$, any nondegenerate code $Q$ subspace of the system $S$ must satisfy
  \begin{equation}\label{eq:packing}
    \dim(S) \ge \dim(Q) \rank(R_{\cal{E}}).
  \end{equation}
  We refer to Eq. (\ref{eq:packing}) as quantum packing bound.
\end{thm}

\begin{proof}
  The thesis follows immediately from Prop. \ref{thm:uncorrelated} and
  \ref{thm:deletion}, along with the considerations above.

  Here we provide an alternative short proof that makes use of more
  technicalities. Diagonalize the matrix $M$ in Eq. (\ref{eq:knill_laflamme})
  to obtain a diagonal matrix $D$ and a new error basis $J_i$. The
  correctability condition in Eq. (\ref{eq:knill_laflamme}) then becomes
  \begin{equation}\label{eq:kldiag}
    P_Q J_i^{\dagger} J_j P_Q = D_{ij} P_Q.
  \end{equation}
  Make use of the polar decomposition and of the correctability condition in
  Eq. (\ref{eq:kldiag}) to obtain
  \begin{equation}
    J_i P_Q = U_i \sqrt{P J_i^{\dagger} J_i P_q} = \sqrt{D_{ii}} U_i P_Q,
  \end{equation}
  where $U$ is some unitary matrix. Thus, the action of the error $J_i$ is to
  rotate $Q$ into the subspace defined by the projector $P_i := U_i P
  U_i^\dagger = J_i P U_i^\dagger / \sqrt{D_{ii}}$. Since such subspaces are
  orthogonal by Eq. (\ref{eq:kldiag}) and are in number of
  $\rank(D)=\rank(M)$, eq. (\ref{eq:bound}) follows. Finally, using the
  non-degeneracy hypothesis, $\rank(M)=\rank(R_{\cal{E}})$ and the statement
  follows.
\end{proof}

Notice that in Eq. (\ref{eq:packing}) only the rank of the Choi-Jamio\l kowski
operator describing the noise is involved, but no assumption on the form of
the noise process affecting the encoding system has been formulated. The
quantum Hamming bound \cite{ekert96}, which holds for noise acting
independently on the encoding systems, can be derived from \eq{eq:packing} as
a particular case.  Actually, if we look at the case of qubits and wish to
correct noise affecting at most $t$ qubits, we consider a basis of error
operators given by products of Pauli matrices involving up to $t$ qubits.
Then, correcting all errors is equivalent to correcting the random-unitary
channel $\cal E$ whose Kraus operators are proportional to the possible
products of $i\le t$ Pauli matrices \cite{nota}.  Since Pauli matrices are
orthogonal, this Kraus representation is already minimal, whence
$\textrm{rank}(R_{\mathcal{E}})$ can be straightforwardly derived counting the
number of independent Kraus operators (the ones affecting $i$ qubits are
$3^i{n \choose i}$) and this leads to
\begin{equation}\label{eq:hamming}
  2^n \ge 2^k \sum_{i=0}^t 3^i {n \choose i}.
\end{equation}
The above formula can be straightforwardly generalized to $q$-dimensional
systems by replacing powers of $2$ with powers of $q$, and 3 with $q^2-1$.

\section{Degenerate codes for correlated noise}\label{sect:examples}

We will now consider the case where noise is correlated, i.e. it does not act
independently on the encoding systems and cannot be expressed as
$\mathcal{E}_1 \otimes \mathcal{E}_2 \otimes \dots \otimes \mathcal{E}_n$,
where $\mathcal{E}_i$ represent the noise process acting on the $i$-th system
in the encoding space.  We will consider in the following the case of qubits
with Pauli correlated noise, i.e. correlated channels in which each Kraus
operator is the product of Pauli matrices \cite{macchiavello02}.

As a first example consider the following CP map $\mathcal{E}$
\begin{eqnarray}\label{eq:2_qubits_correlated_noise}
  \mathcal{E}(\rho) = p \rho + \sum_{i = 1 \dots n, j > i} (p_{X,ij}X_iX_j
  \rho X_jX_i +\nonumber\\ p_{Y,ij}Y_iY_j \rho Y_jY_i + p_{Z,ij}Z_iZ_j \rho
  Z_jZ_i)\;,
\end{eqnarray}
where with probability $p=1-\sum_{i = 1 \dots n, j > i}
p_{X,ij}+p_{Y,ij}+p_{Z,ij}$ the input state is left unchanged, while with
probabilities $p_{X,ij}$, $p_{Y,ij}$ and $p_{Z,ij}$ pairwise Pauli operators
$X$, $Y$ and $Z$ are applied to qubits $i$ and $j$ respectively.  By
evaluating the rank of the above CP map, the quantum packing bound
(\ref{eq:packing}) in this case takes the simple form
\begin{equation}
  2^n \ge 2^k \left[1+3{n \choose 2}\right]\;.
\end{equation}
Let us consider the simple case $k=1$: the smallest integer that satisfies the
above bound is $n=7$. Nevertheless, we can easily construct error correcting
codes with lower values for $n$.  Consider the quantum code spanned by
\begin{equation}
  \ket{\bar0} = \ket{000} \qquad \ket{\bar1} = \ket{111}.
\end{equation}
Notice that the above states are also the codewords that saturate the Hamming
bound for classical error on at most one qubit ($t=1$) \cite{ekert96}.  Our
error correcting strategy works as follows: we encode logical $\ket{0}$ into
$\ket{\bar0}$ and logical $\ket{1}$ into $\ket{\bar1}$.  As mentioned above,
this kind of noise either leaves the qubits unchanged, or acts on two of them
with the same Pauli operator.  Notice that the two codewords are not changed
by the application of the $Z$ matrix on any pair of qubits.  As a consequence
of this, the application of pairwise $Y$ operators gives the same result as
the application of pairwise $X$ operators, namely the code is thus degenerate.
Therefore, we have only to correct errors due to the action of the $X$
operators.  To achieve this, we perform a projective measurement on the
bidimensional subspaces $S_{00} = \span\{\ket{000},\ket{111}\}$, $S_{01} =
\span\{\ket{100},\ket{011}\}$, $S_{10} = \span\{\ket{010},\ket{101}\}$ and
$S_{11} = \span\{\ket{001},\ket{110}\}$.  If the outcome of the measurement is
$00$, noise has not affected any qubit; if the result is $01$, noise has
affected qubits $2$ and $3$; if the result is $10$, noise has affected qubits
$1$ and $3$; if the result is $11$, noise has affected qubits $1$ and $2$.  In
these last three situations, acting with $X$ on the corresponding pair of
qubits gives the original qubits.  As we can see, this code exploits the
invariance of the coding subspace under the action of two $Z$ operators to
allow for perfect error correction while strongly violating the quantum
packing bound.

The above error correcting strategy can be also successfully applied to a
generalization of the correlated noise (\ref{eq:2_qubits_correlated_noise}),
where we can add additional terms involving products of even number of Pauli
operators along the same direction.  For example, it is possible to correct in
the same way errors acting also on four qubits.  In this case, the choice of
the codewords is
\begin{equation}
  \ket{\bar0} = \ket{00000} \qquad \ket{\bar1} = \ket{11111},
\end{equation}
and the error correction is performed in a way similar to the previous one.
As before, this code is highly degenerate because it is invariant under the
application of the product of an even number of $Z$ Pauli operators.  The
error syndrome is discovered by performing a projective measurement on the
subspaces $\span\{\ket{00000},\ket{11111}\}$ corresponding to no errors,
$\span\{\ket{11000},\ket{00111}\}$ and all possible permutations corresponding
to two qubits error, and $\span\{\ket{11110},\ket{00001}\}$ and all possible
permutations corresponding to four qubits error.  Then, error correction is
performed in a similar way to the case discussed before.

In this way we have constructed a degenerate quantum code which violates the
corresponding quantum packing bound for nondegenerate codes
\begin{equation}
  2^n \ge 2^k \left[1+3{n \choose 2}+3{n \choose 4} \right]\;,
\end{equation}
which would require for $k=1$ a minimum $n=14$.  By generalizing this
procedure, we can efficiently correct noise acting on every even number of
qubits.  In fact, the strategy we have provided allows one to correct
correlated errors acting on $2, 4, 6, \dots 2m$ qubits coding on $n=2m+1$
qubits.  The two coding states are then given by
\begin{equation}
  \ket{\bar0} = \ket{0}^{\otimes 2m+1} \qquad \ket{\bar1} = \ket{1}^{\otimes
    2m+1}.
\end{equation}
In this case, the quantum packing bound becomes
\begin{equation}
  2^n \ge 2^k \sum_{i=0}^{m}3{n \choose 2i}
\end{equation}

We emphasize that the possibility of achieving such compact quantum codes for
correlated noise of the form studied here is related to the fact that we
consider error operators acting on an even number of qubits.  We now consider
the problem of correcting correlated noise on three qubits.  The noisy channel
is then of the form
\begin{eqnarray}\label{eq:3_qubits_correlated_noise}
  \mathcal{E}(\rho) = p\rho + \sum_{i = 1 \dots n, j > i, k > j} p_{X,ijk}
  X_iX_jX_k \rho X_kX_jX_i +\nonumber\\ p_{Y,ijk} Y_iY_jY_k \rho Y_kY_jY_i +
  p_{Z,ijk} Z_iZ_jZ_k \rho Z_kZ_jZ_i.
\end{eqnarray}
The smallest number of physical qubits we can employ for such a channel is
$n=3$, which corresponds to a nondegenerate code which saturates the
corresponding quantum packing bound $2^n\geq 2^k (1+3{n \choose 3})$.
Actually, it is possible to encode one logical qubit on $n=3$ physical ones by
employing the additional two qubits as an ancilla initially fixed in the state
$\kket{a} = \ket{0}\ket{+}$, where $\ket{+}=1/\sqrt{2}(\ket{0}+\ket{1})$.
After the action of noise, the two ancilla qubits are measured in the
computational basis and in the $\ket{+}$, $\ket{-}$ basis respectively.  From
the result of this measurement, we learn exactly which of the four Kraus
operator has acted on the qubits because the operators $XX$, $YY$ and $ZZ$
applied to the above state $\kket{a} = \ket{0}\ket{+}$ transform it to the
mutually orthogonal states $\ket{1}\ket{+}$, $\ket{1}\ket{-}$ and
$\ket{0}\ket{-}$ respectively.  In order to correct the errors, after learning
the result of the measurement, we either do nothing or we apply one of the
three Pauli operators on the first qubit to recover the noiseless state.  This
strategy, in contrast to what happens in quantum codes for independent noise,
does not involve multipartite entanglement in the encoding systems as the
three qubits in the encoded state are always factorized.  In this case there
seems to be an intriguing balance between the correlations of the noise and
the entanglement in the encoding system: if the noise is fully correlated on
the three qubits then no entanglement is needed for encoding, while if the
noise is independent then encoding is performed on multipartite entangled
states.

The above procedure can be employed to correct correlated noise of the form
(\ref{eq:3_qubits_correlated_noise}) acting on an arbitrary number $n$ of
qubits by encoding $k=n-2$ qubits.  As before, the $k$-qubit state to be
protected is encoded appending to it the two ancilla qubits in state
$\ket{0}\ket{+}$.  After the receipt of the encoded state, the previously
described measurement of the ancilla will give the syndrome and the necessary
operations to be performed on the rest of the qubits to rescue the original
state. In this case we again construct nondegenerate codes that saturate the
corresponding quantum packing bound.

\section{Conclusion}\label{sect:conclusion}

In this paper we provided a quantum packing bound for nondegenerate codes. The
bound holds for any kind of noise and depends on the rank of the Choi-Jamio\l
kowski operator representing the noise process.  The quantum Hamming bound is
then recovered in the particular case of arbitrary noise acting independently
on a fixed number of encoding systems.  While the quantum Hamming bound has
not been violated so far, in the case of correlated noise we have shown how to
exploit degeneracy to violate the quantum packing bound and achieve perfect
quantum error correction with fewer resources than those needed for
nondegenerate codes.

This work was supported by Italian Ministry of Education through PRIN 2008 and
by the European Community through COQUIT and CORNER projects. Research at
Perimeter Institute for Theoretical Physics is supported in part by the
Government of Canada through NSERC and by the Province of Ontario through MRI.


\begin{thebibliography}{}
\bibitem{shor95} P. W. Shor, Phys. Rev. A {\bf 52}, R2493 (1995).
\bibitem{steane96} A. M. Steane, Phys. Rev. Lett. {\bf 77}, 793 (1996).
\bibitem{gottesman96} D. Gottesman, Phys. Rev. A {\bf 54}, 1862 (1996).
\bibitem{lidar98} D. A. Lidar, I. L. Chuang, and K. B. Whaley,
  Phys. Rev. Lett. {\bf 81}, 2594 (1998).
\bibitem{keyl02} M. Keyl, and R. F. Werner, Springer, Lecture Notes in Physics
  611, 263 (2002).
\bibitem{gregoratti03} M. Gregoratti, and R. F. Werner, Journal of Modern
  Optics, {\bf 50}, 915 (2003).
\bibitem{kribs04} D. Kribs, R. Laflamme, and D. Poulin, Phys. Rev. Lett. {\bf
  94}, 180501 (2005).
\bibitem{NC00} I. L. Chuang and M. A. Nielsen, {\em Quantum Information and
  Communication} (Cambridge, Cambridge University Press, 2000).
\bibitem{AKP06} D. Aharonov, A. Kitaev, and J. Preskill, Phys. Rev. Lett. {\bf
  96}, 050504 (2006).
\bibitem{ekert96} A. Ekert and C. Macchiavello, Phys. Rev. Lett. {\bf 77},
  2585 (1996).
\bibitem{gottesman09} D. Gottesman, arXiv:quant-ph/0904.2557v1 (2009).
\bibitem{sarvepalli08} P. Sarvepalli and A. Klappenecker, Phys. Rev. A {\bf
  81}, 032318 (2010).
\bibitem{ion_traps} J. I. Cirac and P. Zoller, Phys. Rev. Lett. {\bf 74}, 4091
  (1995); A. Garg, Phys. Rev. Lett. {\bf 77}, 964 (1996).
\bibitem{quantum_dots} D. Loss and D. P. Di Vincenzo, Phys. Rev. A {\bf 57},
  120 (1998); M. Thorwart, J. Eckel, and E. R. Mucciolo, Phys. Rev. B {\bf
    72}, 235320 (2005).
\bibitem{solid_state} Y. Makhlin, G. Sch\"on, and A. Shnirman,
  Rev. Mod. Phys. {\bf 73}, 357 (2001); A. Schnirman, Y. Makhlin, and G.
  Sch\"on, Phys. Scr. {\bf T102}, 147 (2002).
\bibitem{VRA99} F. Vatan, V. P. Roychowdhury, and M. P. Anantram, IEEE
  Trans. Inf. Theory {\bf 45}, 1703 (1999).
\bibitem{knill97} E. Knill and R. Laflamme, Phys. Rev. A {\bf 55}, 900 (1997).
\bibitem{CDP10} G. Chiribella, G. M. D'Ariano, and P. Perinotti, Phys. Rev. A
  {\bf 81}, 062348 (2010).
\bibitem{SW02} B. Schumacher and M. D. Westmoreland, Journal Quantum
  Information Processing {\bf 1}, 5 (2002).
\bibitem{kks} D. Kretschmann, D. W. Kribs, and R. W. Spekkens, Phys. Rev. A
  {\bf 78}, 032330 (2008).
\bibitem{nota} Clearly, if there is a channel $\cal R$ that corrects all such
  Pauli errors, then $\cal R$ also corrects their randomization $\cal E$. The
  converse follows from the fact that if $\cal R$ corrects $\cal E$ and ${\cal
    E} = \sum_i {\cal E}_i$, where ${\cal E}_i$ are arbitrary quantum
  operations, then necessarily $\cal R $ corrects ${\cal E}_i$, that is ${\cal
    R} {\cal E}_i (\sigma)= p_i \sigma$ for every $\sigma$ such that ${\rm
    supp} (\sigma) \subseteq Q$.
\bibitem{macchiavello02} C. Macchiavello and G.M. Palma, Phys. Rev. A {\bf
  65}, 050301 (2002).
\end{thebibliography}
\end{document}